\documentclass[preprint]{elsarticle}

\journal{Journal of Artificial Intelligence}

\usepackage{xspace}
\usepackage{microtype}
\usepackage{amsmath}
\usepackage{amsthm}
\usepackage{amssymb}
\usepackage{times}
\usepackage{url}
\usepackage{tikz}
\usepackage{paralist}

\newcommand{\U}{\ensuremath{\mathcal{U}}\xspace}
\renewcommand{\phi}{\varphi}
\newcommand{\limpl}{\supset}
\newcommand{\prefix}{\preceq}
\newcommand{\atom}[1]{\mathtt{#1}}
\newcommand{\seq}[1]{\langle#1\rangle}
\newcommand{\set}[1]{\{#1\}}

\newcommand{\before}{<}

\newcommand{\past}{\Diamond}
\newcommand{\always}{\square}
\newcommand{\Lang}{\ensuremath{\mathcal{L}}\xspace}
\newcommand{\BLang}{\ensuremath{\boldsymbol{\mathcal{L}}}\xspace}

\newcommand{\NP}{NP\xspace}

\newcommand{\iec}{i.e.,\xspace}
\newcommand{\egc}{e.g.,\xspace}
\newcommand{\resp}{resp.,\xspace}

\newtheorem{definition}{Definition}

\newtheorem{theorem}{Theorem}
\newtheorem{lemma}{Lemma}

\begin{document}

\begin{frontmatter}

\title{A Knowledge Representation Perspective on \\ Activity Theory}

\author{Johannes Oetsch and Juan-Carlos Nieves}
\address{		Department of Computing Science\\
	Ume\aa{} University, Sweden\\
	\url{{joetsch,jcnieves}@cs.umu.se}}


%

\begin{abstract}
Intelligent technologies, in particular systems to  promote health and well-being, 
are inherently centered around the human being, and they need to interrelate with human activities at their core. While social sciences provide angles to study such activities, \egc within the framework of cultural-historical activity theory, there is no formal approach to give an account of complex human activities from a 
Knowledge Representation and Reasoning (KR) perspective. Our goal is to develop a  logic-based framework to specify complex activities that is directly informed by activity theory. 
There, complex activity refers to the process that mediates the relation between a subject and some motivating object which in turn generates a hierarchy of goals that direct actions. We introduce a new temporal logic to formalise key concepts from activity theory and study various inference problems in our framework.  We furthermore discuss how to use Answer-Set Programming as a KR shell for activity reasoning that allow to solve various reasoning tasks in a uniform way. 
\end{abstract}


\end{frontmatter}


\section{Introduction}\label{sec:introduction}

The human being is at the heart of intelligent technology like smart assistants or coaching systems
that helps to promote health and well being and facilitates independence for, \egc elderly or 
impaired people.   
Such systems often need to reason about \emph{complex human activities} in different ways. 
This involves tasks like recognising, predicting, or explaining  activities.
An adequately elaborate  notion of activity seems thus fundamental to such  technology. 
%
There is a considerable body of literature on 
recognising plans, intends, and activities of agents~\cite{Sukthankar2014a}. 
Interestingly however, we are  not aware of any formal approach in the literature that attempts to answer the  question ``What is complex human activity?'' 
from a Knowledge Representation and Reasoning (KR) perspective that is 
based on a notion of activity that goes beyond an ad hoc or intuitive understanding--- 
a serious attempt at an answer would be a significant contribution.

A quite detailed picture of the nature of human activity is provided by social science, in particular by  
\emph{cultural-historical activity theory}~\cite{Leontiev1978,Kaptelinin2006c}. 
This theory originates from Soviet psychology as an attempt to study the human mind in a socio-cultural context. 
In a nutshell, activity is understood as a mediated process that relates a subject and an object. 
Activities involve motives and hierarchically structured goals. While the activity is oriented towards the motive, 
actions are targeted at goals. Activities thus take place at different levels: the level of motives, the level of actions and goals, and the operational level where actions are actually executed. 
%
%
%
%

The vision of this work is to develop a more elaborate understanding of complex human activities from a KR perspective that
draws its inspiration directly from cultural-historical activity theory. 
Our motivation to resort to activity theory is to 
enable an domain expert to express knowledge using his or her terminology and within the experts conceptual framework. Indeed, activity theory is 
established to some degree in areas like health care and occupational therapy~\cite{Salo-Chydenius1990}.  
We want to develop a logic-based framework for \emph{activity reasoning}. Notably, this goes beyond plain recognition and involves various other 
practically relevant tasks like prediction, explanation, and meta-reasoning tasks like equivalence and entailment tests.
In particular, we are interested in answering the following questions:
\begin{itemize}
	\item[(Q1)] What activity is currently pursued? Why is it pursued?
	\item[(Q2)] Has the activity been completed? Which goals are relevant to complete it?
	\item[(Q3)] Is it possible to complete the activity? What prevents it from being completed?
	\item[(Q4)] What are potential next goals in order to complete the activity?
	\item[(Q5)] Who needs to cooperate to complete the activity?
	\item[(Q6)] Are two activities equivalent? Does one subsume the other?
\end{itemize}

As the nature of complex human activities is a wide and rather open-ended theme,
we  focus in this work on higher-level aspects of activities like logical and temporal dependencies between objects at different levels of the activity.
We want to express statements about the \emph{achievement of hierarchically structured goals} that
are needed to complete an activity. To this end, we introduce a new temporal logic that 
allows to express sub-goal relations and both logical and temporal dependencies. 
Effectively, this logic is a weaker version of common linear-time temporal logic (LTL)~\cite{Emerson1990} but avoids its high computational costs.
Our idea is strikingly simple: 
 we consider classical propositional logic but interpret formulas over sequences of atoms instead of sets. The syntax is extended by a single modal operator to
exploit this additional order information in the model. Although conceptually  straightforward, this idea has not been studied before and leads to a logic that is sufficiently expressive for our purpose while its computational properties are comparably benign:  The complexity of decision procedures drops from PSPACE for LTL~\cite{Sistla1985} to \NP or even P in our logic. 

We discuss how this logic serves as basis to describe activities as hierarchically structured systems that involve subjects and their goals, mediating tools in both  a material (real object) as well as in an immaterial (\egc skills) sense, and other aspects from activity theory like  the division of labour between different subjects.
We also introduce various reasoning problems that all can be expressed in the new temporal logic.
We thus provide a first formal logic-based definition of complex-human activity from a high-level of abstraction that is directly inspired from activity theory.

We finally use  
Answer-Set Programming (ASP)~\cite{Eiter2009,baral03,Gebser2012}, a prominent approach to declarative problem solving, 
as a KR tool  to realise activity reasoning. 
We illustrate how to specify activity models in a declarative way so that various reasoning tasks, like 
activity recognition, explanation, prediction, and equivalence tests can be solved in a uniform way.
The advantage of a uniform problem encoding is that the activity model can be conveniently specified by facts while the other parts
of the encoding are fixed. Notably, such an encoding is only possible because of the comparatively low computational complexity of our 
new temporal logic.

There is plenty of research on plan and intend recognition that  uses already logic-based frameworks  to
reason about goals and intends of agents that bears similarities to this work.
However, plan and intend recognition are designed to solve a particular reasoning task which is, as the name suggests, recognition. 
As mentioned earlier, we are interested in a broader range of inference modalities, cf.~questions (Q1) to (Q6). 

When it comes to representing goal hierarchies, related work on hierarchical planning~\cite{Gabaldon2009,Blount2012} and Hierarchical Task Network (HTN) planning~\cite{Georgievski2014} come into mind. 
So how do we justify to introduce a new simple logic that may even seem impoverished compared to such highly expressive planning languages?
%
One reason is  that planning frameworks do not provide an elaborate notion of complex human activity. They  
typically focus on the effects of 
actions on states of the world and regards activities  as simple, or sometimes hierarchically structured, sequences of actions to achieve certain goals. 
This intuitive understanding is quite limited in comparison to the notion of human activity from activity theory.
In fact, activities sometimes seem irrational if analysed at the level of directing goals only. The classical example is the beater that scares away animals. The actions are directed at a goal that seems conflicting with the beater's motive of hunting down an animal unless we now that there are other hunters who wait in ambush and share a common motive. The notion of a motive is central in activity theory and not present in planning frameworks.   
 
The other big difference is that hierarchical planing systems provide highly expressive modelling languages.
This is justified as the motivation comes from on how agents and robots should perform tasks in complex environments.
The flip side of this expressiveness are high computational costs.  
While classical planning is highly intractable as even the simplest variants are PSPACE-complete, HTNs are even harder with a complexity anywhere between EXPTIME and undecidable~\cite{Alford2014,Erol1996}. 
Our motivation is however different: We want to model human beings. This indeed justifies a quite high level of abstraction since we do not need to lay out detailed plans for them, and we can
indeed go the opposite direction and reduce both conceptual and computational complexity. This allows for practical reasoning methods, \egc our uniform ASP encoding, that
would not have been feasible otherwise.   
Finally, HTN planners are designed to compute plans while we are interested in different reasoning tasks that also include meta reasoning like
the equivalence of activity models. Due to the lower complexity of our approach, we can achieve this  elegantly using ASP.

Our main contributions can be summarised as follows:
\begin{enumerate}[\rm(i)]
	\item 
	We introduce {\bf a novel temporal logic}. It is conceptually simple but still sufficiently expressive. Modelling is easy as it is close to familiar propositional logic, and  its complexity is not worse than that of classical propositional logic. 
	This is a contribution to temporal logics for itself with potentially  applications beyond the ones discussed here.
	
	\item 
	We introduce {\bf a first formal definition of complex human activity that is inspired directly by socio-cultural activity theory}. It is based on our new logic and thus comes with a clear syntax and semantics.
	 This definition is the basis for various reasoning tasks beyond plain activity recognition. 
	
	\item We introduce {\bf a uniform ASP encoding to practically realise activity reasoning} 
	that allows to solve different tasks in a flexible way. 
	Although the involved problems are hard, ASP solvers have stood the test of time for tackling similar KR problems in the past quite successfully~\cite{Erdem2016}. 
	
\end{enumerate}

This article is organised as follows: 
We first review  activity theory in Section~\ref{sec:background}. Then, we discuss  our motives for a new temporal logic  in Section~\ref{sec:logic}, and introduce its syntax and semantics. We then show how to formally describe activity systems in Section~\ref{sec:formalisation} and how various activity reasoning tasks can be expressed in our logic. In Section~\ref{sec:asp}, we introduce a uniform ASP encoding for activity systems and illustrate how to solve  different inference problems.
Finally, related work is discussed in Section~\ref{sec:relwork}, and we conclude in Section~\ref{sec:conclusion}.

\section{Background on Activity Theory}\label{sec:background}

\begin{figure}
	\centering \small
	\begin{tikzpicture}
	
	\node  (Mediation) at (3,3) { Mediating Artifacts};
	\node  (Subject)   at (1,1) { Subject};
	\node  (Object)    at (5,1) { Object};
	
	\draw [<->,thick] (Mediation) to (Subject);
	\draw [<->,thick] (Mediation) to (Object);
	\draw [<->,thick] (Subject) to (Object);

	\end{tikzpicture}
	\caption{Activity as mediated subject-object relation.
		\label{fig:triangle}
	}
\end{figure}
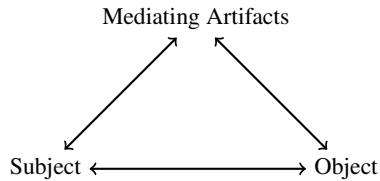

Historically speaking, activity theory is the result of a shift in paradigm among Soviet psychologists in the first half of the 20th century. 
This change has been towards an understanding of the human mind as a result of the interaction between an individual (the subject) and the world (the object) in order to 
satisfy the individual's needs. 
This so-called socio-cultural perspective regards activity as a fundamental unit of life that transforms both the subject and the object and is always mediated through artifacts.  
Activity theory is not a predictive theory but an analytic concept that can act as a lens to study certain problems. Since a focus on the human being in relation to mediating technology  is inherent in activity theory, it
soon became an important foundation in the field of Human Computer Interaction~\cite{Kaptelinin2006c}
and is, in general, used interdisciplinary in fields beyond psychology.  
Activity is understood as 
relation between a subject (the human being) and an object (a goal that the subject wants to achieve) that
is mediated by tools and situated within a social context. This relation is illustrated in Fig.~\ref{fig:triangle}. 
Our understanding of activity is based on the work of Aleksei Leontiev~\cite{Leontiev1978} who originally introduced activity theory as a conceptual framework. 
We follow Kaptelinin and Nardi~\cite{Kaptelinin2006c} in the presentation of the following principles: 

\paragraph{Object-orientedness}
Activity is understood as the process that relates a subject and an object. 
Hence, the subject is always acting towards some object and activity is defined by this relation that transforms both the subject and the object.
For us, subject will refer to an individual although in general this
could refer to a group or organisation as well. The object is something that exists in the world. It is what the subject wants to obtain, the goal of the activity. 
Objects are not restricted to the physical world, they can  be intangible  as well as long as their existence is objective.
For example, in the activity of learning for a medicine exam, knowledge about the human body could be the object that has to be obtained.    

\paragraph{Hierarchical structure}
Activities can be decomposed into different layers. The top-layer object of an activity is its motive. The motive  is the ultimate goal that drives the activity, it
is why the activity takes place and is closely related to the subject's needs. Then, there are a goals subordinated to the  motive that need to be achieved in order to complete the activity and that direct the subject's actions. These goals 
are hierarchically structured themselves and  
consist of subgoals at different levels. The subject's actions are consciously directed towards the goals. The third level  accommodates operations, which describe how actions are implemented by lower-level routines under certain conditions. 
The distinction between motivating and directing objects is a feature of human activity  that  is in particular important when considering activities in a social context that involves \emph{division of labour}. It is quite common that overall motivating objects are shared in a group while individuals act towards different subordinated directing goals.      

\paragraph{Mediation}
Another key concept in activity theory is that a subject never acts directly towards an object. The process is always mediated by artifacts like
material or immaterial tools. 
Often, such mediating tools encode knowledge and experience of previous generations that is passed down from one generation to another. Any activity is shaped by the tools that are used and, vice versa, activities shape
the tools they involve. 
Mediating means also include acquired skills or knowledge which can be viewed as internalised tools.       
Generally speaking, internalisation and externalisation refer to a transfer from parts of an activity from the external world to within the subject and and the other way round. 

\smallskip
Other aspects of activities, like the  role of rules and norms as well as the  dialectic approach of activity theory for development, are beyond the scope of this brief introduction. Also note that we do not deal with organisations which are more prominent in the Scandinavian brand of activity theory\cite{Engestroem1987}. From a computer-science perspective, Leontiev’s approach, which is predominantly concerned with activities of individuals, was more appealing.
To sum up, complex human activity is understood as the process that relates an individual with some objective goal. This goal, dubbed the motive, defines the activity and generates a hierarchy of goals that direct the actions of the individual. We want to provide a formal definition of activity that takes these ideas into account. In particular, we will present a logic-based approach to represent and reason with goal hierarchies that are linked to tools and motives of individuals.

\section{A Logic for Abstract Activity Reasoning}\label{sec:logic}

Our goal is to develop a formalism to specify complex activities that support various reasoning tasks that are relevant to implement intelligent systems with the human being in the loop. 
We focus on rather high-level aspects of activities and present a simple propositional temporal logic for this purpose.
Before introducing syntax and semantics, 
we discuss our motivation  for a new logic. 

\subsection{Motivation}\label{sec:motivation}

We will use the 
traditional Swedish activity of having \emph{fika}---a coffee break and social institution---as a running example. 
Fika involves having a coffee and eating a cinnamon bun or a sandwich. 
Moreover, we assume that we need to open the fridge and get a plate before we can eat a sandwich. Likewise, 
we need to open the cabinet and get a plate before we can eat a cinnamon bun.

The activity of having fika can be analysed at different levels according to activity theory. 
Approaches that use machine learning or probabilistic methods based on sensor data will often focus on the operational level 
as they observe the effects of operations of an individual. 
Thus, they recognise actions and their associated goals rather than activities according to our understanding.
Although this kind of aggregation and interpretation of data is an important intermediate step, it often fails
to recognise the bigger picture. For example, we could learn that a person is heating water from sensor data. However, the more meaningful  broader activity of preparing a soup (or defrosting a lock in winter, etc.) could be missed. Also, since there is no explicit model of the activity, it is not easy to generate explanations for observations or predictions in data-driven approaches.

Planning frameworks, on the other hand, are well suited to describe activities at the level of actions and goals.  
But as goals are the effects of (sequences of) actions, such frameworks are inherently unable to make the distinction between motivating objects, \iec motives, and
directing goals like in the example from the introduction. 
We are interested in expressing dependencies between objects of the activity but abstract away the actual low-level actions to achieve them.

For example,  a decomposition of the fika activity into goals and subgoals
appears in Fig.~\ref{fig:and-or-graph} as and-or graph. 
That two arcs are connected by an arc means that all involved subgoals need to be obtained first. If there is no arc, just one of the subgoals is required to achieve the goal. 
\begin{figure}
	\centering
	\begin{tikzpicture}
	
	\node (fika)        at (3,4)  {fika};
	\node (coffee)      at (1,3)   {coffee};
	\node (sandwich)    at (7,2)   {sandwich};
	\node (cinnamonbun) at (4,2)   {cinnamon\_bun};
	\node (cabinet)     at (3,1)   {cabinet};
	\node (plate)       at (6,1)   {plate};
	\node (fridge)      at (8,1)   {fridge};
	\draw [-latex] (fika)    to (coffee);
	\draw (fika)    to (5,3);
	\draw [-latex](5,3) to (sandwich);
	\draw [-latex](5,3) to (cinnamonbun);
	\draw [-latex](cinnamonbun) to (cabinet);
	\draw [-latex](cinnamonbun) to (plate);
	\draw [-latex](sandwich)    to (plate);
	\draw [-latex](sandwich)    to (fridge); 
	\path (fika) -- coordinate[near start] (A) (coffee);
	\path (fika) -- coordinate[near start] (B) (5,3);
	\draw (A) to[bend right] (B);
	
	\path (cinnamonbun) -- coordinate[near start] (C) (cabinet);
	\path (cinnamonbun) -- coordinate[near start] (D) (plate);
	\draw (C) to[bend right] (D);
	
	\path (sandwich) -- coordinate[near start] (E) (plate);
	\path (sandwich) -- coordinate[near start] (F) (fridge);
	\draw (E) to[bend right] (F);
	
	\end{tikzpicture}
	\caption{Goals and subgoals in the fika example.
		\label{fig:and-or-graph}
	}
\end{figure}
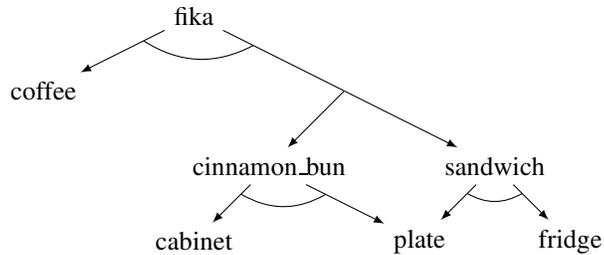

\paragraph{Design considerations}
We  identify a couple of design requirements for specifying  goal structures as in Fig.~\ref{fig:and-or-graph}:     
There are logical relations like disjunctions and conjunctions of subgoals. Subgoals can also be understood as  implications: a goal can be obtained only if the subgoal has been achieved. However, being a subgoal  also carries temporal information: the subgoal has to precede the goal in time. There will, in general, be further logical and temporal dependencies between goals that need to be taken into account.

Another important design consideration are the reasoning task that we want to support. This is a crucial point: we are not interested in mere activity recognition but in 
\emph{activity reasoning} which encompassed a variety of inference problems that are distilled from practical requirements. In particular, if we assume that we obtain information about goals that have been achieved or that ar being pursued by methods that are situated at the operational level or the level of actions, 
we would like to be able to answer the following questions (reiterated from the introduction):
\begin{itemize}
	\item[(Q1)] What activity is currently pursued? Why is it pursued?
	\item[(Q2)] Has the activity been completed? Which goals are relevant to complete it?
	\item[(Q3)] Is it possible to complete the activity? What prevents it from being completed?
	\item[(Q4)] What are potential next goals in order to complete the activity?
	\item[(Q5)] Who needs to cooperate to complete the activity?
	\item[(Q6)] Are two activities equivalent? Does one subsume the other?
\end{itemize} 


Given all the design considerations, temporal logics seem to be a natural choice as they  allow
to reduce various inference problems to logical entailment in a uniform way. However, common temporal logics like LTL~\cite{Pnueli1977}  are costly: 
deciding satisfiability is PSPACE-complete for the prominent LTL versions~\cite{Sistla1985,Lichtenstein1985}. 
We introduce a new  propositional temporal logic, dubbed logic \Lang, that basically extends classical propositional logic by a single modal operator.
This logic is simple yet sufficiently expressive and has comparably good computational properties.

\paragraph{Atomic goal atoms}
The atomic unit in \Lang are \emph{goal atoms} which are either true (the goal has been achieved) or false (a goal has not been achieved).
This is  a more abstract perspective on goals than in, \egc action theories where they are seen as certain desirable states of the  world and are described by sets of fluents. 
A goal atom expresses that a given goal has been achieved without implying anything else about the state of the world. 
For example, ``empty glass on table'' can be seen as a goal atom which is achieved if someone puts an empty glass on the table.
Later, this glass  could be filled with water or put away. Still,  the goal atom remains true
, even if the glass is neither empty nor on the table as it expresses an achievement and not a state of the world.

\paragraph{Persistence  and the Heraclitus principle}
Any goal is achieved at a particular time instant,  and we  thus
assume an instant-based model of time. Although, for example, drinking a cup of coffee  is also a process of a certain time span, and such processes may overlap with others, the associated
goals are reached only at specific points in time, \egc at the end of a process.
Since we understand a goal as something that once achieved cannot be taken away, a distinctive feature of our logic is that
goals persist: once a goal atom becomes true, it remains true forever. 
Directly related to this persistence of goal achievements is the principle that a goal cannot be obtained twice.  
We follow here the Heraclitus principle stating that:
``No man ever steps in the same river twice, for it is not the same river and he is not the same man''.
Well, but what if one decides to have one more cup of coffee? Then, it will not be the same cup of coffee  and thus a new goal that is distinguishable from the first one.
In practice, this often means no real limitation: Suppose we want to have, e.g., two cups of coffee instead of one. We can simply introduce a second distinguishable goal ($\atom{coffee}_2$ besides $\atom{coffee}_1$).

\subsection{Syntax of Logic \BLang}
We consider an instance of a propositional linear-time logic with the following past operations~\cite{Emerson1990}:

$\begin{array}{ll}
\always\phi & \mbox{it has always been the case that\ } \phi, \\  
\past\phi & \mbox{it has at some time been the case that\ } \phi\mbox{, and} \\
\phi\before\psi & \phi \mbox{ has been the case before\ } \psi.    
\end{array}$

\medskip
\noindent
Temporal  logic \Lang  is formally defined as follows:

\begin{definition}\label{def:syntax}
	Let \U be a countable set of goal atoms, 
	language \Lang  over  \U  is the smallest set such that
	\begin{enumerate}[\rm(i)]
		\item for any goal atom $a \in \U$, $a \in \Lang$;
		\item if $\phi \in \Lang$, then $(\neg \phi) \in \Lang$;
		\item if $\phi,\psi \in \Lang$, then $(\phi \land \psi) \in \Lang$; and 
		\item if $\phi \in \Lang$, then $(\always \phi) \in \Lang$.
	\end{enumerate}
\end{definition}
As usual, we might omit parentheses a long as no ambiguities arise, and we will often leave \U implicit.
We also use the following abbreviations:
\[
\begin{array}{rcl@{\quad}rcl}
\phi \lor \psi& := &   \neg(\neg \phi \land \neg \psi)\,, &  
\phi \limpl \psi& := & \neg \phi \lor \psi\,, \\
\past \phi& := &\neg \always \neg \phi\,, &
\phi \before \psi& := & \psi \limpl \past \phi\,. 
\end{array}
\]

We assume the following precedence order among the connectives:
\[
\always,\past, \before;  \qquad \neg; \qquad \land,\lor; \qquad \limpl.
\]


\subsection{Semantics of Logic \BLang}

We start by  establishing some notation.
We use $\seq{a_1,\ldots,a_n}$ to denote a finite {sequence} of elements, while sets are written as $\set{a_1,\ldots,a_n}$.
For two sequences $P$ and $Q$, $PQ$ is the concatenation of $P$ and $Q$.
We say that $P$ is a prefix of $Q$ 
and write $P \prefix Q$ 
iff $Q = PR$, for some sequence $R$.   
The empty sequence is denoted by $\varepsilon$. Note that the notion of prefix is reflexive.

In classical propositional logic, a model is commonly represented by the set of atoms that are true in that model.
In \Lang, we will use sequences of atoms instead of sets. 
\begin{definition}\label{def:model}
	An \emph{ordered model} over \U is a sequence of atoms from \U without repetition. 
\end{definition}

Intuitively, an ordered model represents in which order certain goals are reached.
The fact that we do not allow for repetition reflects the Heraclitus principle.
For example,
\[
\seq{\atom{cup}, \atom{coffee}, \atom{sandwich}}
\] 
describes the scenario where we first get a cup, then we drink a coffee, and finally we eat a sandwich.
As discussed above, the process of eating a sandwich and drinking a coffee can be overlapping,  but reaching the goal of drinking a cup of coffee and eating a sandwich are
regarding as specific points in time where one has to precede the other.

In fact, this ordered model is a compact representation of a sequence of possible worlds, each of which can be
represented as a classical model, that describe how the world evolves from one state to another:
\[
\emptyset, \{\atom{cup}\}, \{\atom{cup}, \atom{coffee} \}, \{\atom{cup}, \atom{coffee}, \atom{sandwich} \}\,. 
\]

\begin{definition}\label{def:models}
	Given an ordered model $M$ and a formula $\phi \in \Lang$ over \U, we write
	$M \models \phi$ to denote that $\phi$ is \emph{true} in $M$.
	The truth of a formula in \Lang is inductively defined as follows:
	\begin{enumerate}[\rm(i)]
		\item $M \models a$ iff $a$ occurs in $M$, for an atom $a \in \U$;
		\item $M \models \neg \psi$ iff $M \not\models \psi$;
		\item $M \models \psi \land \phi$ iff $M \models \psi$ and $M \models \phi$;
		\item $M \models \always \phi$ iff $M' \models \phi$ for each prefix $M'$ of $M$;
	\end{enumerate}	
\end{definition}
The concepts of of theory, entailment, equivalence, etc. are defined as in classical logic.	
As the notion of prefix is reflexive, $\always \phi$ means that $\phi$ is true now and at any time in the past.
Likewise,  $\past \phi$ means that $\phi$ is true now or at some time in the past.

We are able to give a formal specification of the fika example from Fig.~\ref{fig:and-or-graph}:
\[
\begin{array}{l}
\always((\atom{coffee} \land (\atom{cinnamon\_bun} \lor \atom{sandwich})) \before \atom{fika}) \\
\always((\atom{fridge} \land \atom{plate}) \before \atom{sandwich}) \\
\always((\atom{cabinet} \land \atom{plate}) \before \atom{cinamon\_bun})
\end{array}
\]
An ordered model that satisfies all theses formulas would be \[
\seq{\atom{plate},\atom{fridge},\atom{coffee}, \atom{sandwich}, \atom{fika}}\,.\]
Note that $<$ is usually combined with $\always$: compare
$\seq{a,b,c} \models b \before a$ with $\seq{a,b,c} \not\models \always(b \before a)$. 

Recall that goal atoms have the property that once they become true, they remain true. This concept generalises to formulas as follows.
\begin{definition}\label{def:goalformulas}
	A formula $\phi$ has \emph{positive persistence} iff $M \models \phi$ implies $M' \models \phi$, for any ordered models $M$ and $M'$ with $M \prefix M'$. 
	Conversely, $\phi$ has \emph{negative persistence} iff $M \not\models \phi$ implies $M' \not\models \phi$, for any ordered models $M$ and $M'$ with $M \prefix M'$.
	
	If a formula has positive persistence, we refer to it as a \emph{goal formula}.
\end{definition}
It can be checked that $(\atom{cinnamon\_bun} \lor \atom{sandwich})$ from the fika example is indeed a goal formula.

Note that for any formula $\phi$, $\always \phi$ has negative persistence. Also,  the negation of a formula with positive, \resp negative, persistence, has
always negative, \resp positive persistence. This implies that any formula of form $\past \phi$ is a goal formula.

As usual, we say that an occurrence of a subformula is positive, \resp negative, if it is in the scope of an even, \resp odd, number of negation 
symbols. We proceed with a further syntactic characterisation of goal formulas.
\begin{theorem}\label{thm:goalformulas}
	Let $\phi$ be a formula,  
	\begin{enumerate}[\rm(i)]
		\item if the polarity of each atomic subformula of $\phi$ is positive, then $\phi$ is a goal formula;
		\item if the polarity of each atomic subformula of $\phi$ is negative,  then $\phi$ has negative persistence.
	\end{enumerate}
\end{theorem}
\begin{proof}
	By structural induction on $\phi$. If $\phi$ is atomic, it is obviously a goal formula. 
	If $\phi = \psi' \land \psi''$ and (i) applies, $\psi'$ and $\psi''$ are both goal formulas by induction, and thus
	$\phi$ is a goal formula as well. If (ii) applies,  $\psi'$ and $\psi''$ have negative persistence and so does $\psi' \land \psi''$.
	If $\phi = \always \psi$, we only need to consider Case~(i) as negative persistence is implied for any formula of this form. As $\psi$ is a goal formula by induction,  so is  $\always \psi$. 
	Finally, assume $\phi = \neg \psi$. Without loss of generality, assume that Case~(i) applies. But then Case~(ii) holds for $\psi$ and it is by induction a formula of negative persistence. As the negation of a formula of negative persistence yields a formula of positive persistence, $\neg \psi$ has to be a goal formula. 
\end{proof}

We next formalise the idea of goal hierarchies and define a sub-goal relation.
Intuitively, a subgoal is characterised by two properties: first,  it is necessary to achieve the higher-level goal, and second, it is achieved first.
\begin{definition}\label{def:subgoal}
	Let $T$ be a set of formulas. Formula $\psi$ is a \emph{subgoal} of formula $\phi$ in $T$ iff
	$T \models \always(\psi \before \phi)$. 
\end{definition}

Note that the fika example is a collection of explicit subgoal definitions.
An example of an implicit subgoal relation is that $\atom{coffee} \land \atom{plate}$ is a subgoal of $\atom{fika}$.
In case the considered formulas are goal formulas, the characterisation of subgoals simplifies:
\begin{theorem}\label{thm:subgoals}
	For formula $\phi$ and goal formula $\psi$, 
	$\always(\psi \before \phi)$ is equivalent to $\always(\phi \limpl \psi)$.
\end{theorem}
\begin{proof}
	Recall that $\always(\psi \before \phi)$ is defined as ($\ast$) $\always(\phi \limpl \past \psi)$.
	As $\psi$ entails $\past \psi$ for any formula $\psi$, and
	$\past \psi$ entails $\psi$ for any goal formula $\psi$, $\past \psi$  can be replaced by  $\psi$ in ($\ast$) and the equivalence follows.
\end{proof}
As all subgoal relations in the fika example are defined between goal formulas, we could apply the previous result and replace
all occurrences of $\before$ by implications.

Ordered models prove to be robust against irrelevant information in their representation:
\begin{theorem}\label{thm:langrep}
	Let $M$ be an ordered model, $\phi \in \Lang$, and $M'$ the 
	subsequence of $M$ that consists of  precisely the  atoms in $M$ that occur in $\phi$.
	Then, $M \models \phi$ iff $M' \models \phi$. 
\end{theorem}
\begin{proof}
	The proof proceeds by structural induction on $\phi$. 
	%
	If $\phi$ is an atom, then 
	$M \models \phi$ iff $M' \models \phi$ as $\phi$ occurs in $M$ iff $\phi$ occurs in $M'$.
	
	
	Assume $\phi = \neg \phi'$.  
	By our induction hypothesis, $M \models \phi'$ iff $M' \models \phi'$. Thus $M \models \neg \phi'$ iff $M' \models \neg \phi'$.   	
	
	Assume  $\phi = \phi' \land \phi''$.
	Then, $M \models  \phi' \land \phi''$ iff $M \models  \phi'$ and $M \models  \phi''$.
	By induction, $M \models \phi'$ iff $M' \models \phi'$, and $M \models \phi''$ iff $M' \models \phi'$.
	Hence,  $M \models  \phi' \land \phi''$ iff $M' \models \phi' \land \phi''$.
	
	Assume $\phi = \always \phi'$, and $M' \models \always \phi'$.
	Hence, $N' \models \phi$, for any prefix $N'$ of $M'$.
	Let $N$ be an arbitrary prefix of $M$. The subsequence of $N$ that contains all atoms from $N$ that also occur in $\phi$
	is a subsequence of $M'$.    
	The induction hypothesis implies  $N \models \phi$.
	Therefore, $M \models \always \phi$.
	Conversely, assume $M  \models \always \phi'$. It follows that $N \models \phi'$, for any prefix $N$ of $M$.
	For an arbitrary such prefix $N$,
	let $N'$ be the prefix of $M'$ that contains all atoms of $N$ that also occur in $\phi'$.
	By induction, $N \models \phi'$ only if $N' \models \phi'$, 
	and $M' \models \always \phi$ follows. 
\end{proof}
Note that the above theorem would fail if our logic contained or allowed to express a previous 
time operator that is common in many other temporal logics.

Language \Lang is in fact a weaker but for our purposes sufficiently expressive version of LTL~\cite{Pnueli1977} 
with only past operators~\cite{Lichtenstein1985}. 
Models of linear time logics are usually infinite sequence of states and deciding satisfiability is PSPACE-complete for
the prominent LTL versions~\cite{Sistla1985,Lichtenstein1985}. 
Due to the compact representation of ordered models in \Lang, complexity of reasoning tasks considerably drops, which was a big motivation to introduce a dedicated new logic in the first place. 

\begin{lemma}\label{lem:poly}
	Deciding whether a formula $\phi \in \Lang$ is satisfiable by an ordered model $M$ is feasible in polynomial time.
\end{lemma}
\begin{proof}
	We construct a table, where the rows are labeled with all the subformulas of $\phi$, and the columns are labeled by all prefixes of $M$.  This table is then filled row by row with truth values that correspond to the truth of the respective subformula in the considered prefix. We start with atoms and proceed from less complex formulas to more complex ones.
	The truth values of formulas that are not atomic is determined by inspecting the truth values of their subformulas for the prefixes prescribed by Definition~\ref{def:models}. Formula $\phi$ is satisfiable iff the truth value in the row for $\phi$ and column for $M$ is true. 
	This table is constructible in polynomial time.
	We illustrate this construction in Table~\ref{tab:sat}.
\end{proof}

This  result is in stark contrast to the \NP-hardness of determining truth of a propositional LTL formula in a 
finite structure even  when restricted to the single temporal operator for ``eventually''~\cite{Sistla1985}.

\begin{lemma}\label{lem:hardness}
	Deciding whether a formula $\phi \in \Lang$ is satisfiable is \NP-hard.
\end{lemma}
\begin{proof}
	We reduce  
	the satisfiability problem  in classical propositional logic to the satisfiability problem in \Lang.
	Let $M$ be a set of atoms and $\phi$ a formula of classical logic.
	Then, $M$ satisfies $\phi$ in classical logic iff  $M' \models \phi$ in  \Lang, where $M'$ is any sequence of the atoms in $M$.
\end{proof}

\begin{theorem}\label{thm:lang}
	Deciding whether a formula $\phi \in \Lang$ is satisfiable is \NP-complete.
\end{theorem}
\begin{proof}
	We can decide if a  formula $\phi$ is satisfiable by nondeterministically guessing an ordered model $M$ over the atoms in $\phi$ (Theorem~\ref{thm:langrep}) and checking in polynomial time (Lemma~\ref{lem:poly}) whether $M \models \phi$. 
	Lemma~\ref{lem:hardness} shows hardness.
\end{proof}	

\begin{table}[t]
	\centering
	\caption{Deciding if $\seq{a,b,c} \models (a \lor b) \before c$ holds.}
	\label{tab:sat} 
	
	\begin{tabular}{l@{\qquad}c@{\qquad}c@{\qquad}c@{\qquad}c@{\qquad}}
		& $\varepsilon$ & $\seq{a}$ & $\seq{a,b}$ & $\seq{a,b,c}$ \\ \hline
		$a$   						& f & t & t & t \\
		$b$   						& f & f & t & t \\
		$c$   						& f & f & f & t \\
		$a \lor b$					& f & t & t & t \\
		$(a \lor b) \before c$ 		& t & t & t & t \\
		
	\end{tabular}	
	
\end{table}

\section{Formalising Activity Systems}\label{sec:formalisation}

In this section, we discuss how to represent and reason about complex human activities and
relate our formalisation to the activity theory concepts from Section~\ref{sec:background}.
We hence provide a first formal logic-based definition of an activity system that is inspired directly from activity theory. We then
study the different  inference problems from Section~\ref{sec:motivation}.


\subsection{Activity Systems}

The following definition makes it precise how to use logic \Lang
to describe what we call an activity system.  

\begin{definition}\label{def:activity} 
	An \emph{activity system} is a tuple $(S,O,M,A)$, where
	$S$, $O$, and $M$, are countable sets, the \emph{subjects}, \emph{objects}, and \emph{mediating artifacts} of the activity system, respectively, 
	and $A$ is a set of activities. An \emph{activity} is a  
	triple $(m,G,T)$ with $\U = S \times O$, and
	\begin{enumerate}[\rm(i)]
		\item $m \in \U$ is the \emph{motive of the activity};
		\item $G \subseteq \Lang$ is a set of \Lang formulas over \U; and 
		\item $T: \U \mapsto 2^M$ maps  goal atoms to sets over $M$.  
	\end{enumerate}
\end{definition}

Since activity is understood as a mediated subject-object relation, Definition~\ref{def:activity} involves three sorts of entities: 
the \emph{subjects} $S$, the \emph{objects} $O$, and mediating artifacts $M$.
All goal atoms are subject-object pairs from   $S \times O$ which reflects the concept of object-orientedness from activity theory.

The motive $m$ is the top-level object of the activity. It describes \emph{why} the activity takes place and is, according to the 
hierarchical structure of activities, a motivating object that corresponds to {needs} of an individual.
The set $G$ of formulas describes which goals have to be obtained to complete the activity associated with the motive. These directing goals are at a lower level in the hierarchical structure. 
Similar to the motive, the directing goals are subject-object pairs that describe \emph{what} needs to be achieved by which individuals. 
The directing goals of an activity may involve different subjects which is a convenient way to model  division of labour and social context.

Finally, function $T$ defines sets of mediating tools that are required for an individual to obtain an object. The idea is that for each subject-object pair $(s,o) \in G$, at least one set in $T(P)$ of tools is required to be available for $s$ when acting towards $o$. Tools may be material and immaterial ones, and also skills are seen as internalised artifacts that mediate an activity. Mapping $T$ can be specified concisely by associating each goal atom $a \in G$ with a formula $\phi$ of classical logic over $M$ such that $T(a)$ is given by the models that satisfy $\phi$. 
An activity system is simply understood as a collection of activities with common subjects, objects, and tools.

Let us illustrate a simple activity system that consists of only one activity based on the fika example.
The motive $m$ of this activity is $(ebba,fika)$ which can be interpreted as Ebba wants to have fika. This goal in turn corresponds to needs of Ebba like social engagement, motivates the activity, and explains why  actions take place.
The goal structure $G$ for $m$ is specified by the following formulas:
\[
\begin{array}{l}
\atom{(ebba,coffee)} \land  (\atom{(ebba,cinnamon\_bun)}  \lor \atom{(ebba,sandwich)})\,, \\ 
\always((\atom{(elsa,fridge)} \land \atom{(ebba,plate)})  \before \atom{(ebba,sandwich)})\,, \\
\always((\atom{(elsa,cabinet)} \land \atom{(ebba,plate)})  \before \atom{(ebba,cinamon\_bun)})\,.
\end{array}
\]

Note that some of the above goals need to be achieved by Ebba's helping assistant Elsa instead of Ebba. In this activity system, we assume that Ebba is mildly impaired and is unable to do certain tasks like going to the fridge or opening the cabinet due to, \egc mobility issues. 
Some goals are predicated upon mediating artifacts. For example, the process of archiving $g = (\atom{ebba},\atom{coffee})$ requires the use of a coffee machine as a physical tool, skills to operate the machine, and is mediated through social rules how to drink the coffee: $T(g) = \set{\set{\atom{c\_machine},\atom{c\_skills},\atom{fika\_etiquette}}}$.

\subsection{Reasoning Problems}
Activity reasoning involves, besides the representation of activities, reasoning tasks that are distilled from
practical requirements like recognition, explanation, and prediction  as discussed Section~\ref{sec:motivation}.
We next discuss these inference tasks in more detail.

\paragraph{Completion}
The most basic task is to ask which activities have been completed. 
In more formal terms, given an ordered model $I$, an activity system
$\mathcal{A} = (S,O,M,A)$, and a motive $m$ from an activity in $A$, \emph{completion} is the problem of deciding whether
$I \models G$ for some activity $a \in A$ with $a = (m,G,T)$. 

We say that an atom $p$ is \emph{relevant} to complete activity $a$, if $p$ occurs in $I$ and $I' \not\models G$, where $I'$ is the subsequence of $I$ that does not contain $p$.

\paragraph{Achievability}
Relevant for predicting activities and reasoning about limitations of available tools is to determine which activities can 
be completed in the future.
Let $I$ be an ordered model, 
$t: \U \mapsto M$ a specification of the currently available tools for subject-object pairs,
$\mathcal{A} = (S,O,M,A)$ an activity system, and $m$ a motive from an activity in $A$. 
\emph{Achievability} is the problem of deciding if there exists an ordered model $I'$ with $I \prefix I'$ , such that
$I' \models  G$ for some activity $a \in A$ with $a = (m,G,T)$, and for each goal atom $g$ in $I'$, for some element $e \in T(g)$, $e \subseteq t(g)$.
While completion can be decided in polynomial time, cf.~Lemma~\ref{lem:hardness}, achievability is a  harder problem:
\begin{theorem}\label{thm:real}
	Deciding achievability is NP-complete.
\end{theorem}
\begin{proof}
	We nondeterministically guess ordered model $I'$ over the language of atoms that occur in $G$ from any activity $a \in A$ and check the remaining 
	conditions of achievability in polynomial time which shows membership.    
	Hardness follows from a simple reduction from the satisfiability problem of formulas in \Lang, cf.~Theorem~\ref{thm:lang}, to achievability.
\end{proof}

Completion and achievability are relevant to recognise finished and ongoing activities as well as to yield \emph{explanations} by inferring potential motives. 
Also, the task of \emph{predicting} potential next goals towards the completion of an activity can be formulated as a achievability problem:
Assume motive $m$ is realisable in the activity system $\mathcal{A}$ with respect to the ordered model $I$. We can decide if $g \in \U$ is a potential 
next goal by deciding achievability for the ordered model $I\seq{g}$.
To identify \emph{abnormalities}, models can be investigated in a  backward fashion as well: Assume we know that motive $m$ drives an activity in $\mathcal{A}$ but
$m$ is not realisable in the current ordered model $I$. The longest prefix of $I$ that yields achievability of $m$, if such a prefix exists, will point
to the subsequent goal in $I$ that prohibits completion of $m$.
A simple variation of the achievability problem is \emph{social achievability} which additionally specifies a set $J$ of subjects. A motive is deemed realisable under this notion if motive $m$ is realisable in an ordered model that only involves subjects from $J$. This can be used, for example, to check if Ebba can complete her fika activity without assistance after the goal $(\atom{elsa},\atom{fridge})$ has been achieved.  
In a similar way, questions about required  mediating artifacts to complete activities can be addressed by abductive reasoning from $t$.

\paragraph{Entailment and equivalence of activities}
An important meta-reasoning task is to determine if one activity is entailed by another activity or if two activities are 
equivalent. This is relevant to simplify the representation of activity systems by getting rid of redundant conditions. Fortunately, 
as \Lang is a monotonic logic, these tasks can be realised in a straight-forward manner by directly using the logic's entailment operator.
For activity system $\mathcal{A}$ with activities $a_1 = (m_1,G_1, T_1)$ and $a_2 = (m_2,G_2, T_2)$, $a_1$ entails $a_2$ iff $G_1 \models \phi$, for each
formula $\phi \in G_2$. We say that $a_1$ and $a_2$ are equivalent iff  $a_1$ entails $a_2$ and $a_2$ entails $a_1$.
In the fika example,  the simpler formula 
$
\always(\atom{(ebba,sandwich)} \limpl  
(\atom{(elsa,fridge)} \land \atom{(ebba,plate)}))
$
is entailed by the activity that is formalised in the fika example (cf.~Theorem~\ref{thm:subgoals}).

\section{ASP for Activity Reasoning}\label{sec:asp}

An important facet of activity reasoning is how to practically realise it.  
Our goal is not to discuss an implementation for a particular task. Rather, 
we want to describe activity systems 
so that they function as a 
knowledge base for KR systems and support  
various inference problems in a uniform way
Our tool of choice is \emph{Answer-Set Programming} (ASP).

ASP is a prominent approach for declarative problem solving with numerous applications in AI and KR related areas~\cite{Erdem2016}. 
The roots of ASP lie in logic-based knowledge representation, nonmonotonic reasoning, and logic programming 
based on the 
stable model semantics~\cite{gelfond88,gelfond91}. 
The success is due to simple, yet expressive, modelling languages along with efficient solvers like 
Clingo, DLV, or WASP.\footnote{\url{https://potassco.org}, \url{http://www.dlvsystem.com}, \url{http://alviano.github.io/wasp}}
The idea of ASP as a problem solving paradigm is to declaratively specify a problem, called the \emph{ASP program}, and to use an ASP solver to compute solutions, dubbed the \emph{answer sets} of the 
program. We refer the reader to related work~\cite{baral03,Eiter2009,Gebser2012} for a proper introduction and provide only intuitive explanations in the remainder of this section. 

\subsection{Fact Representation of Activities}
\begin{figure*}
	\centering	
	{\small
		\begin{verbatim}
		% formulas associated with the motive (ebba,fika)
		formula((ebba,fika), and((ebba,coffee),
		     or((ebba,sandwich),(ebba,cinamon_bun)))).
		formula((ebba,fika), h(before(
		     and((elsa,fridge),(ebba,plate)),(ebba,sandwich)))).
		formula((ebba,fika), h(before(
		     and((elsa,cabinet),(ebba,plate)),(ebba,cinamon_bun)))).
		\end{verbatim}
	}
	\caption{Fact representation of the fika example.}\label{fig:facts}
\end{figure*}
We first illustrate how activity systems can be modelled in the input language of the solver Clingo. A \emph{fact representation} of the fika example
appears in Fig.~\ref{fig:facts}.
Predicate \verb|formula(M,F)| specifies that motive \verb|M| is associated with formula \verb|M|.
Formulas are written in prefix notation using term symbols \verb|and|, \verb|or|, \verb|impl|, and \verb|neg| for the classical connectives. 
For the temporal operators $\always$, $\past$, \resp $\before$, we use 
\verb|h|, \verb|p|, \resp \verb|before|. Subjects, objects, directing as well as motivating goals are implicit. The aspect of mediating artifacts is left out for the sake of a more clear and succinct presentation.
It is also possible to specify prefixes $\seq{a_1,\ldots,a_n}$ of ordered models as part of this input using facts 
\verb|prefix(|$1$\verb|,|$a_1$\verb|)|, \ldots, \verb|prefix(|$n$\verb|,|$a_n$\verb|)|. We will discuss applications later on.
The fact representation of \Lang formulas and prefixes can be regarded as the input for the ASP solver.

\subsection{Uniform Problem Encoding}
The next part of the ASP model, given in Fig.~\ref{fig:guess},  are rules that span the search space in terms of ordered models. 
In ASP terminology, this is the \emph{guessing part}. While the fact representation differs for different activities,
the guessing part
is fixed. 
\begin{figure}
	{\small
		\begin{tabular}{l}
			\verb|% subformula relation for formulas|  \\	
			\verb|subformula(S) :- formula(_,S).|  \\
			\verb|atom((S,O))   :- subformula((S,O)).| \\
			\verb|subformula(S) :- subformula(|$U$\verb|(S)).|  \\
			\verb|subformula(S) :- subformula(|$B$\verb|(S,_)).|  \\
			\verb|subformula(S) :- subformula(|$B$\verb|(_,S)).|  \\
			~~~for $U \in \{\verb|neg|, \verb|h|, \verb|p| \}$  and  $B \in \{ \verb|and|, \verb|or|, \verb|impl|, \verb|before| \}$ \\ 
		\end{tabular}
		
		\begin{verbatim}
		% guess length of the ordered model
		minlen(N) :- #count{ A: prefix(P,A),atom(A) } = N.
		maxlen(M) :- #count{ X: atom(X) } = M.
		1 {length(X): X =  N..M} 1 :- minlen(N), maxlen(M).
		
		% guess an ordered model of that length
		model(N,A) :- prefix(N,A), atom(A).
		1 {model(P,A): atom(A),not prefix(A,_)} 1 :- 
		  P = (M+1)..N, minlength(M), length(N).
		:- model(P,A), model(P1,A), P < P1.
		\end{verbatim}
	}
	\caption{Guessing ordered models.}\label{fig:guess}
\end{figure}

We use predicate \verb|models/2| to represent ordered models:
any ordered model  $\seq{a_1,\ldots,a_n}$ corresponds to the set of atoms
\verb|model(|$1$\verb|,|$a_1$\verb|)|, \ldots, \verb|model(|$n$\verb|,|$a_n$\verb|)| and vice versa.
An ASP solver, \egc Clingo, can readily be used to produce all ordered models 
(or all models that contain a specified prefix) over the language of input formulas. 
This  alone is of course not very useful. We need to relate ordered models with truth values of the inputs formulas.

The rules in Fig.~\ref{fig:define} define the truth values for each subformula in the ordered model generated by the guessing part. 
Again, this part is fixed for any input.
We can run Clingo on the entire encoding to inspect the ordered models  together with the truth values of all subformulas in these models. 
\begin{figure}
	{\small
		\begin{verbatim}
		% prefix indices of the ordered model
		index(N) :- length(M), N = 0 .. M. 
		
		% truth values of formulas  
		t(S) :- formula(_,S), t(S,N), length(N).
		f(S) :- formula(_,S), not t(S).
		
		% truth value of atoms
		t(A,N) :- atom(A), model(N1,A), N1 = 1..N, index(N).  
		f(A,N) :- atom(A), not t(A,N), index(N). 
		
		% truth value of not
		t(F,N) :- subformula(F),F=neg(F1),f(F1,N).
		f(F,N) :- subformula(F),F=neg(F1),t(F1,N).
		
		% truth value of or
		t(F,N) :- subformula(F),F=or(F1,F2),t(F1,N).
		t(F,N) :- subformula(F),F=or(F1,F2),t(F2,N).
		f(F,N) :- subformula(F),F=or(F1,F2),f(F1,N), f(F2,N).
		
		% truth value of and
		t(F,N) :- subformula(F),F=and(F1,F2),t(F1,N),t(F2,N).
		f(F,N) :- subformula(F),F=and(F1,F2),f(F1,N).
		f(F,N) :- subformula(F),F=and(F1,F2),f(F2,N).
		
		% truth value of impl
		t(F,N) :- subformula(F),F=impl(F1,F2),f(F1,N).
		t(F,N) :- subformula(F),F=impl(F1,F2),f(F2,N).
		f(F,N) :- subformula(F),F=impl(F1,F2),t(F1,N),f(F2,N).
		
		% before(F,G) 
		t(F,N):-subformula(F),F=before(F1,F2),index(N),t(F1,N1), N1<=N.
		t(F,N):-subformula(F),F=before(F1,F2),index(N),f(F2,N).
		f(F,N):-subformula(F),F=before(F1,F2),index(N),not t(F,N).
		
		% p(F) 
		t(F,N) :- subformula(F), F = p(F1), index(N), t(F1,N1), N1 <= N.
		f(F,N) :- subformula(F), F = p(F1), index(N), not t(F,N).
		
		% h(F) 
		f(F,N) :- subformula(F), F = h(F1), index(N), f(F1,N1), N1 <= N.
		t(F,N) :- subformula(F), F = h(F1), index(N), not f(F,N).
		\end{verbatim}
	}
	\caption{Truth values of subformulas in an ordered model.}\label{fig:define}
\end{figure}

The intuition that the ASP encodings of Fig.~\ref{fig:guess} and \ref{fig:define} characterise the semantics of logic \Lang can be formalised as follows.
\begin{theorem}\label{thm:aspencoding}
	Given a set $G$ of \Lang formulas over a set $A$ of atoms and an ordered model $P$ over $A$.
	The answer sets of  \[\mathit{P} = \mathit{Guess} \cup  \mathit{Define} \cup \mathit{Formulas} \cup \mathit{Prefix}\] are in one-to-one correspondence
	with the ordered models of $G$ that have $P$ as prefix, where
	\begin{itemize}[--]
		\item $\mathit{Guess}$ are the rules from Fig.~\ref{fig:guess},
		\item $\mathit{Define}$ are the rules from  Fig.~\ref{fig:define},
		\item $\mathit{Formulas}$ is the fact representation of formulas $G$, and
		\item $\mathit{Prefix}$ is the fact representation of prefix $P$.	
	\end{itemize}
	
	Furthermore, let $S$ be an answer set of $\mathit{P}$ that represents an ordered model $M$, then
	$S$ contains atom $\mathtt{t(F)}$, if the formula that corresponds to term $\mathtt{F}$ is true in $M$, and  $\mathtt{f(F)}$ otherwise. 
\end{theorem}
\begin{proof}[Proof (Sketch)]
	The role of the rules in Fig.~\ref{fig:guess} is to generate
	answer sets that corresponds to ordered models over $A$
	with $P$ as prefix.
	Predicate \verb|subformula/1| provides an  auxiliary definition that is used later on: atom \verb|subformula(F)| is derived iff the formula corresponding to $F$ is a subformula of any of the input formulas. Predicate \verb|atom/1| is used to denote atomic subformulas. 
	
	The minimal length of an ordered model, represented by \verb|minlen/1|, is the number of atoms in the prefix $P$ (0 if no such prefix is specified as input).
	The maximal length, given by \verb|maxlen/1|, is the number of atoms in the inputs formulas.
	The actual length of the model, given as  \verb|length/1|, is nondeterministically selected to be a number between the minimal and the maximal length.
	
	The ordered model is represented by atoms 
	\verb|model(N,A)|, meaning that atom \verb|A| occurs at position \verb|N| of the model, where \verb|N| ranges from $1$ to the length of the model. 
	For any position in the sequence of the considered length that is not already determined by the prefix $P$, one atom from the input formulas is nondeterministically selected. Furthermore, a constraint ensures that different positions of the model cannot hold the same atom.

	Atom \verb|t(F,N)| is derived iff the subformula associated with $F$ is true under the prefix of length $N$ of the ordered model that is defined by \verb|model/2|. Likewise, \verb|f(F,N)| is derived if that formula is false.
	For every subformula and any length between 0 (the empty prefix) and the length of the entire ordered model, either a respective atom \verb|t/2| or \verb|f/2| is derived.
	We hence follow closely the membership argument in the proof of Lemma~\ref{lem:poly} and work in a bottom-up fashion from the atoms to more complex subformulas.  
	The actual rules that derive this truth values follow the semantics of the logical and temporal operators from Definition~\ref{def:syntax}.
	
	Finally, \verb|t(F)| is derived iff \verb|t(F,N)| is true, \verb|F| is an input formula, and \verb|N| is the length of the entire ordered model. Atom  \verb|f(F)| is true iff \verb|t(F)| is not true.
\end{proof}

\subsection{ASP for Activity Reasoning}

The problem encoding serves as a basis to realise different reasoning tasks. 
We illustrate how to solve them in ASP using our running example.

\paragraph{Completion}
To infer which activities have been completed in a given ordered model, we can 
specify a model using predicates \verb|prefix| and \verb|length| as additional input facts and run the ASP solver with the added rules:
{\small
	\begin{verbatim}
	-completed(M) :- formula(M,F), f(F).
	completed(M) :- formula(M,_), 
	not -completed(M).  
	\end{verbatim}} 
\noindent
The solver will return a single answer set that contains atom \verb|completed(M)| for any motive \verb|M| that has all its formulas satisfied by the specified ordered model.
If we specify an ordered model as follows:
{\small \begin{verbatim} 
	length(3).
	prefix(1,(elsa,fridge)).
	prefix(2,(ebba,plate)).
	prefix(3,(ebba,sandwich)).
	\end{verbatim}}
\noindent
we will get one answer set that does not contain any atom \verb|completed/1| as the goal $(\atom{ebba},\atom{coffee})$ is missing to satisfy the fika model.

\paragraph{(Social) achievability and prediction}
We also can partially specify an ordered model and check if an activity can be completed in an extension of that  model, thus
solving a problem of achievability.
To test if the motivating goal $\atom{(ebba,fika)}$ can be realised after the goals $\atom{(ebba,coffee)}$ and  $\atom{(elsa,fridge)}$ have been achieved, we could
add
{\small\begin{verbatim}
	prefix(1,(ebba,coffee)). 
	prefix(2,(elsa,fridge)).
	:- -completed((ebba,fika)).
	\end{verbatim}} 
\noindent
to the ASP program. The last constraint prunes away all answer sets that correspond to ordered models which do not satisfy all formulas associated with the motive.

Possible next goals to realise $\atom{(ebba,fika)}$  can be inferred by computing the brave consequences (atoms true in all answer sets, Clingo option \verb|-e brave|) of the program which reveals that either $\atom{(ebba,plate)}$ or $\atom{(elsa,cabinet)}$ need to become true next. This corresponds to prediction problems from the previous section.

Tasks of social achievability can be solved in a similar way. For instance, we can add the constraints
{\small\begin{verbatim}
	:- -completed((ebba,fika)).
	:- true((S,O)), S <> ebba. 
	\end{verbatim}}  
\noindent
to  test if Ebba can have fika without assistance, which is not the case (Clingo will output no answer sets). 

\paragraph{Relevance and minimal ordered models}
We might find it irritating that ordered models will also contain atoms that we deem irrelevant for realising an activity. A remedy is to stipulate that models are of minimal length
which can easily expressed with an optimisation statement
{\small\begin{verbatim}
	#minimize { N,N: length(N) }.
	\end{verbatim}}  
\noindent
that adds a penalty of $N$ to any model  of length $N$.
Now inspecting the answer sets shows that in a minimal ordered model, $\atom{(ebba,plate)}$ is the only possible next goal as $\atom{(elsa,cabinet)}$ is not
necessary to complete the activity.

\paragraph{Entailment and equivalence}
The meta-reasoning task of deciding whether an activity entails another one or if two activities are equivalent can be expressed in ASP quite neatly.
Remember that the  before operator can  be replaced by an implication according to Theorem~\ref{thm:subgoals}.
Hence, we would expect that the input formulas entail, for example, the formula 
$\always(\atom{sandwich} \limpl  \atom{fridge})$. 
To verify this, we add the following rules to the encoding:
{\small\begin{verbatim}
	formula((ebba',fika'),
	h(impl((ebba,sandwich),(elsa,fridge)))).
	
	fail :-  completed((ebba,fika)), 
	-completed((elsa',fika')).
	:- not fail.
	\end{verbatim}}  
Atom \verb|fail| will be derived iff the solver finds an counterexample to the entailment problem, 
that is, an ordered model that satisfies all input formulas but fails to satisfy the new formula. The last rule eliminates all
answer sets that do not correspond to a counterexample. Hence, an ASP solver will return no answer set as entailment holds.
As equivalence is defined as the conjunction of two entailment problems, we can decide it  by two separate  calls.

\section{Related Work}\label{sec:relwork}

Interest in human activities from an KR perspective has a long history and early work on activity representations based on scripts dates back to the 1970s~\cite{Schank1977}. 
This approach, further discussed in  Winograd and Flores classic book Understanding Computers and Cognition~\cite{Winograd1987}, also supports a variety of reasoning tasks but was not pursued within a logical paradigm.

There is already a considerable body of literature when it comes to recognising activtities. 
We refer the reader to related work for an overview and snapshot of  the field of plan, activity, and intent recognition~\cite{Sukthankar2014a}.
Some work uses machine learning to analyse sensor data and thus operates more on the operational level. 
We discussed the limitations of such a data-driven approach when it comes to higher-level reasoning tasks and explanations in Section~\ref{sec:logic}.
Other higher-level goal recognition approaches use probabilistic or statistical methods and are thus different from this work
that is based on qualitative reasoning exclusively. 
A logic-based approach based on event-calculus to 
recognise human activities from video content has been put forth by Artikis et al.~\cite{Artikis2013}.
Gabalon~\cite{Gabaldon2009} and Blount and Gelfond~\cite{Blount2012}
also use ASP for recognising intended actions and hierarchically structured activities in the context of planning. 
It is important to stress that  we deal with the broader problem of activity reasoning while
goal and intend recognition is more narrowly focused on one particular inference problem.   


Different approaches that allow to reason about intentions and beliefs of agents based on modal logic and the BDI architecture  
have been studied for intelligent agents~\cite{Wooldridge2000}. As the focus of that work are artificial agents, 
an explicit notion of complex \emph{human} activity is not present.   

More closely related to our goal of representing hierarchical goal structures is HTN planning~\cite{Georgievski2014} which we discussed already in the introduction in more detail.
The main difference to our approach is that we provide 
(i)   an explicit definition of human activity,
(ii)  a high-level of abstraction justified by modelling the human being instead of artificial agents,
(iii) a simple modelling language with better computational properties, and
(iv)  a flexible solving approach for various reasoning tasks beyond planning.

Temporal logic with past operations for expressing properties of observations and preferences can be found in
other  work on plan and intend recognition~\cite{Sohrabi2016,Sohrabi2011}. 
An interval-based temporal logic has also been used for plan recognition in intelligent help systems~\cite{Bauer1994}, but
inference is based on expensive Prolog-style proof procedures.
A particular feature that distinguishes 
our approach from related work based on planning or other temporal logics are the comparatively low computational costs
of our temporal logic. 

ASP has been used for different applications that involve reasoning about human activities, \egc to implement the intelligent home monitoring system SINDI~\cite{Mileo2009} in the health-care domain.
There, ASP is used to reason about the health development of patients based on 
a context model and sensor networks.
Human activities have also been studied through the lens of argumentation~\cite{Nieves2013,Guerrero2017,Guerrero2018}, where ASP is used as reasoning engine as well. 
ASP has also been used for indoor positioning as part of assistive living environments~\cite{Mileo2011}.  
Do \emph{et al.}~\cite{Do2012} aim at recognising basic and complex activities from smartphone sensors. Basic activities are 
directly recognised from sensor data, \egc running, while complex ones are derived.
Recent work deals with extensions of ASP to support temporal operations in the input language and multi-shot ASP for computing temporal models~\cite{Cabalar2018}.

\section{Conclusion}\label{sec:conclusion}

The vision of this work is to develop an 
understanding of complex human activities from a KR perspective that
is informed by cultural-historical activity theory. 
The need to support a broad range of inference problems in activity reasoning as opposed to mere activity recognition motivates  a 
new light-weight temporal logic that serves as backbone for
high-level descriptions of activities with a focus on
subject-object relations, their hierarchical structure, and mediation.
This logic is effectively a weaker version of standard LTL so that the complexity of decision procedures for relevant reasoning tasks drops to \NP, which makes it especially appealing compared to related logic-based approaches.
  
We furthermore showcase ASP as a flexible KR shell for activity reasoning. The comparatively low complexity of the novel temporal logic allow for a uniform ASP encoding
so that  
we can deal with a wide range of inference problems in a declarative and flexible way. This includes recognition, prediction, and explanation tasks.  Also,  
we can harness the power of readily available solver technology which is an important consideration when designig intelligent systems with the human being in the loop.

A practical evaluation of our approach is outside the scope of this work and remains as a project for future work. 
Reactive aspects of activity reasoning in productive systems by, \egc exploiting stream reasoning in ASP~\cite{Gebser2013,Do2012}, will also be  an interesting future direction when heading towards real-world applications.




\end{document}